\documentclass{journalA4}

\usepackage{amsmath}

\graphicspath{{figures/}}

\title{Discretized Approaches to Schematization%
\thanks{Part of the results appeared in the author's PhD thesis \cite{Meulemans-2014} and supersede the results in arXiv:1306.2827 \cite{Meulemans-2013}. In particular, it provides a much simpler proof for Theorem~\ref{thm:main}, allowing for stronger implications.}
} 

\author{Wouter Meulemans%
\thanks{giCentre, City University London, London, United Kingdom. \texttt{wouter.meulemans@city.ac.uk}}
}

\newcommand{\etal}{{et al.}}
\newcommand{\f}{Fr\'echet }
\newcommand{\C}{\ensuremath \mathcal{C}}
\newcommand{\bd}{\ensuremath \partial}
\newcommand{\df}{\ensuremath d_\textrm{F}}
\newcommand{\ds}{\ensuremath d_\textrm{SD}}
\newcommand{\eps}{\ensuremath \varepsilon}
\newcommand{\R}{\ensuremath\mathbb{R}}

\newtheorem{problem}{Problem}
\newtheorem{corollary}{Corollary}

\newcommand{\myparNS}[1]{\noindent{\sffamily\bfseries #1.}}
\newcommand{\mypar}[1]{\smallskip\myparNS{#1}}

\begin{document}

\maketitle

\begin{abstract}
To produce cartographic maps, simplification is typically used to reduce complexity of the map to a legible level. 
With schematic maps, however, this simplification is pushed far beyond the legibility threshold and is instead constrained by functional need and resemblance.
Moreover, stylistic geometry is often used to convey the schematic nature of the map.
In this paper we explore discretized approaches to computing a schematic shape $S$ for a simple polygon $P$.
We do so by overlaying a plane graph $G$ on $P$ as the solution space for the schematic shape.
Topological constraints imply that $S$ should describe a simple polygon.
We investigate two approaches, \emph{simple map matching} and \emph{connected face selection}, based on commonly used similarity metrics.

With the former, $S$ is a simple cycle $C$ in $G$ and we quantify resemblance via the \f distance.
We prove that it is NP-hard to compute a cycle that approximates the minimal \f distance over all simple cycles in a plane graph $G$. 
This result holds even if $G$ is a partial grid graph, if area preservation is required and if we assume a given sequence of turns is specified.

With the latter, $S$ is a connected face set in $G$, quantifying resemblance via the symmetric difference. 
Though the symmetric difference seems a less strict measure, we prove that it is NP-hard to compute the optimal face set.
This result holds even if $G$ is full grid graph or a triangular or hexagonal tiling, and if area preservation is required.
Moreover, it is independent of whether we allow the set of faces to have holes or not.
\end{abstract}

\section{Introduction}

Cartographic maps are an important tool for exploring, analyzing and communicating data in their geographic context.
Effective maps show their main information as prominently as possible.
Map elements that are not relevant to the map's purpose are abstracted or fully omitted.
In \emph{schematic maps}, the abstraction is taken to ``extreme'' levels, representing complex geographic elements with only a few line segments.
In addition to highlighting the main aspects of the map, they are useful to avoid an ``illusion of accuracy'' \cite{Monmonier-1996} which may arise when showing data on a detailed map: the schematic appearance acts as a visual cue of distortion, imprecision or uncertainty.
However, the low complexity must be balanced with recognizability.
Correct topological relations and resemblance hence play a key role in schematization.
Schematic maps also tend to be stylized by constraining the permitted geometry.
Orientations of line segments are often restricted to a small set $\C$, so called \emph{$\C$-oriented schematization}.
The prototypical example is a schematic transit map (e.g. the London Tube Map), in which all segments are horizontal, vertical or a $45$-degree diagonal.

\begin{figure}[t]
  \centering
  \includegraphics{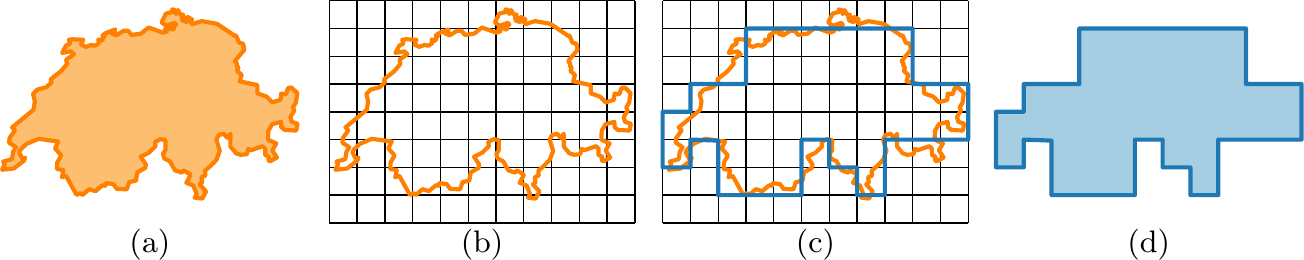}
  \caption{Discretized schematization. (a) Simple polygon $P$ to be schematized (Switzerland). (b) A $\C$-oriented (rectilinear) graph $G$ is placed on $P$. (c) A \emph{simple} polygon $S$ with its boundary  constrained to $G$ that resembles $P$. (d) Shape $S$ is a $\C$-oriented schematization of the input.}
  \label{fig:idea}
\end{figure}

A central problem in schematization is the following: given a simple polygon $P$, compute a simple $\C$-oriented polygon $S$ with low complexity and high resemblance to $P$. 
Typically, one is constrained to optimize the other.
Formalizing ``high resemblance'' requires the use of similarity measures, each having its own benefits and weaknesses \cite{BuchinMRS-2015}.
In this paper we investigate a discretized approach to schematization.
To this end we overlay a $\C$-oriented plane graph $G$ on $P$, and require the boundary of $S$ to coincide with a simple cycle in $G$ (Fig.~\ref{fig:idea}).
Though it restricts the solution space, this approach also offers some benefits. 
\begin{itemize}
\item The graph can easily model a variety of constraints, possibly varying over the map, or mixing with other types of geometry, such as circular arcs.
\item Discretization promotes the use of collinear edges and provides a uniformity of edge lengths. This provides a stronger sense of schematization and a more coherent ``look and feel'' when schematizing multiple shapes.
\item Combined with the simplicity constraint, the graph enforces a minimal width for narrow strips in $P$, avoiding an undesired visual collapse (see Fig.~\ref{fig:exaggeration}).
\end{itemize}
Finally, discretization may be necessary due to the intended use of the schematic shape. 
Examples include computing tile-based cartograms \cite{CanoBCPSS-2015} and deciding for a grid map \cite{EppsteinKSS-2013,SlingsbyDW-2010} on a connected set of cells that resemble the complete region.
These use cases require us to bound not (only) the bends of the result, but the number or total size of the enclosed faces. 
This is also relevant in the context of area-preserving schematization \cite{BuchinMRS-2015}.

\begin{figure}[t]
\centering
\includegraphics{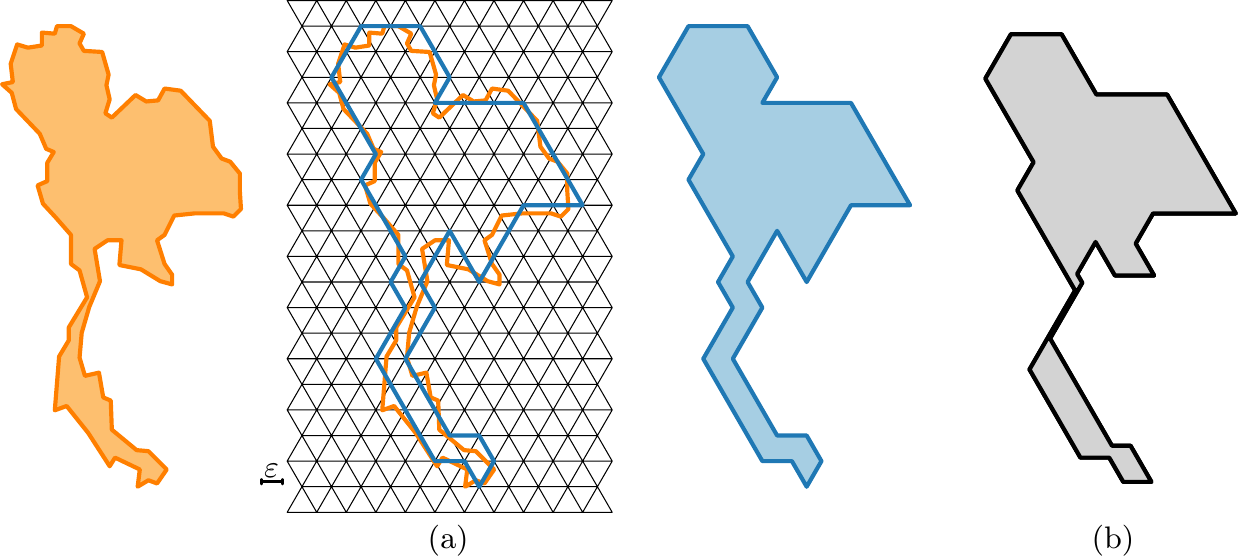}
\caption{(a) A nonintersecting schematization of Thailand with \f distance at most $\varepsilon$ exaggerates the narrow strip. Computed using techniques described in \cite{Meulemans-2014}. (b) Result of algorithm by Buchin~\etal~\cite{BuchinMRS-2015} contains a visual collapse.}
\label{fig:exaggeration}
\end{figure}

We consider two discretized approaches, as outlined below.
Note that we focus on grid graphs (plane graphs with horizontal and vertical edges only) without constraining the complexity of the result.
Refer to Section~\ref{sec:prelims} for precise definitions used in this paper.

\mypar{Simple map matching}
The first approach aims to find a simple cycle in $G$, quantifying resemblance via the \f distance ($\df$).
This leads us to the problem statement below.

\begin{problem}[Simple map matching]
\label{prob:SMM}
Let $G$ be a partial grid graph, let $P$ be a simple polygon and let $\eps > 0$.
Decide whether a simple cycle $C$ in $G$ exists such that $\df(C,P) \leq \eps$.
\end{problem}

In Section~\ref{sec:proof} we prove that this problem is NP-complete. 
This proof has implications on several variants of this problem (Section~\ref{sec:corollaries}). 
In particular, no reasonable approximation algorithm exists, unless P$=$NP.

\mypar{Connected face selection}
With the second approach we use the symmetric difference ($\ds$) to quantify resemblance.
We now consider the full polygon rather than only its boundary, and look for a connected face set in $G$ instead of a cycle. 
This leads to the following problem: 

\begin{problem}[Connected face selection]
\label{prob:CFS}
Let $G$ be a full grid graph, let $P$ be a simple polygon and let $D > 0$.
Decide whether a connected face set $S$ in $G$ exists such that $\ds(S,P) \leq D$.
\end{problem}

In spite of the symmetric difference being insensitive to matching different parts between polygons, we prove that this is again an NP-complete problem in Section~\ref{sec:cfs}.
This result is independent of whether we allow holes in the resulting polygon.
Moreover, the proof readily implies hardness for regular tilings using triangles or hexagons.

\subsection{Related work}

\myparNS{Schematization}
Line and network schematization (e.g. transit maps) have received significant attention in the algorithmic literature, e.g. \cite{CabelloBK-2005,DellingGNPI-2014,FinkHNRSW-2013,NoellenburgW-2011}. 
Recently, schematization of geographic regions has gained increasing attention, e.g. \cite{BuchinMRS-2015,DijkGHMS-2014-giscience,GoethemMSW-2015}.
Our discretized approach is similar in nature to the octilinear schematization technique of Cicerone and Cermignani~\cite{CiceroneC-2012}, though simplicity is of no concern in their work.
As mentioned in the introduction, the discretized approach offers conceptual advantages over the existing nondiscretized methods.

\mypar{Map matching}
Map matching has various applications in geographic information systems, such as finding a driven route based on a road network and a sequence of GPS positions.
Without its simplicity constraint, it has been studied extensively with various criteria of resemblance, e.g. \cite{AltERW-2003,BrakatsoulasPSW-2005,HaunertB-2012,WenkSP-2006}.
Alt~\etal~\cite{AltERW-2003} describe an algorithm that solves nonsimple map matching under the \f distance in $O(mn \log n)$ time where $m$ is the complexity of $P$ and $n$ the complexity of $G$. 
Though ``U-turns'' can be avoided, no general simplicity guarantees are possible.
Similarly, the decision problem for the weak \f distance can be solved in $O(mn)$ time~\cite{BrakatsoulasPSW-2005}.
The \emph{simple} map matching problem on the other hand has received little attention, although it stands to reason that for many applications a nonselfintersecting result is desired, if the input curve is simple.
A full grid graph always admits a solution with Hausdorff distance at most $3\sqrt{2}/2$ and \f distance at most $(\beta + \sqrt{2})/2$, where $\beta$ parametrizes a realistic input model \cite{BoutsKKMSV-2016a,BoutsKKMSV-2016b}. 
Wylie and Zhu~\cite{WylieZ-2014} prove independently that simple map matching under the \emph{discrete} \f distance is NP-hard, however, without requiring a simple input curve.
A stronger result---with a simple input curve---follows directly from our proofs.
Sherette and Wenk \cite{SheretteW-2013} show that it is NP-hard to find a simple curve with bounded \f distance on a 2D surface with holes or in 3D, but again without requiring a simple input curve.

\mypar{Face selection}
In its dual form, connected face selection is a specialization of the known NP-hard maximum-weight-connected-subgraph problem \cite{ElKebirK-2014,Johnson-1985} in which a (planar) graph must be partitioned into two disjoint components, $S$ and $S'$, such that $S$ is connected and has maximal total weight.
Our results readily imply that this dual problem remains NP-hard even on a full grid graph if all weights are nonnegative and the size of $S$ is given; this is independent of whether we constrain $S'$ to be connected.
It is also related to the well-studied graph cuts, though focus there lies minimizing the number of cut edges connecting $S$ and $S'$ (e.g. \cite{FeldmannW-2015}).
Vertex weights have been included only as part of the optimization criterion (e.g. \cite{AroraRV-2009,ParkP-1993}), the other part still being the number of cut edges. 
The number of cut edges is not correlated to the complexity of the eventual shape: we cannot use this as a trade-off between complexity and resemblance.
Also, these approaches tend not to require that a partition is connected, and focus on ensuring a certain balance between the two partition sizes. 
Partitioning an unweighted nonplanar graph into connected components that each contain prescribed vertices is known to be NP-hard \cite{HofPW-2009}.

\subsection{Preliminaries}
\label{sec:prelims}

\myparNS{Polygons}
A polygon $P$ is defined by a cyclic sequence of vertices in $\R^2$.
Each pair of consecutive vertices is connected by a line segment (an edge).
A polygon is \emph{simple} if no two edges intersect, except at common vertices.
We use $|P|$ to refer to the area of polygon $P$.
The \emph{complexity} of a polygon is its number of edges.
We use $\bd P$ to refer to the boundary of $P$.
Unless mentioned otherwise, polygons are assumed to be simple throughout this paper.

\mypar{Graphs}
A straight-line graph $G = (V,E)$ is defined by a set of vertices $V$ in $\R^2$ and edges $E \subset V \times V$ connecting pairs of vertices using line segments.
The graph is \emph{plane} if no two edges intersect, except at common vertices.
The \emph{complexity} of a plane graph is its number of edges.
We call a plane graph a \emph{(partial) grid graph} if all of its vertices have integer coordinates, and all edges are either horizontal or vertical, having length at least $1$. 
A \emph{full} grid graph is a maximal grid graph (in terms of both vertices and edges) within some rectangular region; all edges have length $1$.
A full grid graph represents a tiling of unit squares.
Unless mentioned otherwise, graphs are assumed to be plane in this paper.

\mypar{Cycles}
A \emph{cycle} in a graph is a sequence of vertices such that every consecutive pair as well as the first and last vertex in the sequence are connected via an edge.
A cycle is \emph{simple} if the sequence does not contain a vertex more than once.
A simple cycle in a plane graph corresponds to (the boundary of) a simple polygon.
The \emph{bends} of a cycle are the vertices at which the interior angle is not equal to $\pi$, that is, those that form corners in the polygon it represents.
The \emph{complexity} of a cycle is its number of bends.

\mypar{\f distance}
The \f distance quantifies the (dis)similarity between two geometric shapes: a high \f distance indicates a low similarity.
We define $B_P \colon S^1 \rightarrow \bd P$ as the continuous function that maps the unit circle onto the boundary of $P$.
Let $\Psi$ denote the set of all orientation-preserving homeomorphisms on $S^1$.
Using $\| \cdot \|$ to denote the Euclidean distance, the \f distance between two polygons is defined as
$\df(P,Q) = \inf_{\psi \in \Psi} \max_{t\in S^1} \| B_P(t) - B_Q(\psi(t)) \|.$

\mypar{Faces}
Maximal empty regions of a graph (i.e., not containing a vertex in its interior) are referred to as \emph{faces}.
One face, the outer face, is infinite; all other faces have bounded area.
Faces are said to be adjacent if they share at least one edge. 
The dual $G^*$ of graph $G$ is defined by a vertex set containing a dual vertex for every face and two dual vertices are connected via a dual edge if the corresponding faces are adjacent.
A face set in $G$ is said to be \emph{connected} if the corresponding induced subgraph of $G^*$ is connected.
We call a face set \emph{simply connected} if the complement of the face set is also connected.
A simply connected face set corresponds to a simple polygon; a (nonsimply) connected face set may have holes.

\mypar{Symmetric difference}
The symmetric difference between two polygons $P$ and $Q$ is defined as the area covered by precisely one of the polygons. 
That is, the symmetric difference is 
\[\ds(P,Q) = |(P \cup Q) \backslash (P \cap Q)| = |P \cup Q| - |P \cap Q| = |P| + |Q| - 2 \cdot |P \cap Q|.\] 
If $Q = Q_1 \cup \ldots \cup Q_k$ is composed of pairwise disjoint regions $Q_i$ (e.g. a face set), then we may decompose the above formula into $|P| + \sum_{i=1}^k (|Q_i| - 2 |P \cap Q_i|)$.

\section{Simple map matching is NP-complete}
\label{sec:proof}

In this section we consider Problem~\ref{prob:SMM}, simple map matching.
We prove that this problem is NP-complete, as formalized in the following theorem.
\begin{theorem}\label{thm:main}
Let $G$ be a partial grid graph, let $P$ be a simple polygon and let $\eps > 0$.
It is NP-complete to decide whether $G$ contains a simple cycle $C$ with $\df(C,P) \leq \eps$.
\end{theorem}
The problem is in NP since the \f distance can be computed in polynomial time \cite{AltG-1995} and it is straightforward to check simplicity.
In this section we prove that the problem is also NP-hard.
We conclude this section by considering the implications of this result on variants of this problem.
We assume $\eps = 1$ in the remainder of this section.

\subsection{Reduction overview}

De Berg and Khosravi~\cite{BergK-2010} prove that \emph{planar monotone 3SAT} is NP-complete. That is, decide whether a 3CNF formula $F$ is satisfiable, given the following constraints (Fig.~\ref{fig:planar3sat}(a)):
\begin{itemize}
\item clauses are either positive (only unnegated literals) or negative (only negated literals);
\item a planar embedding for $F$ is given, representing variables and clauses as disjoint rectangles;
\item variables lie on a single horizontal line;
\item positive clauses lie above the variables, negative clauses below;
\item links connecting clauses to the variables of their literals are strictly vertical.
\end{itemize}

\begin{figure}[h]
\centering
\includegraphics{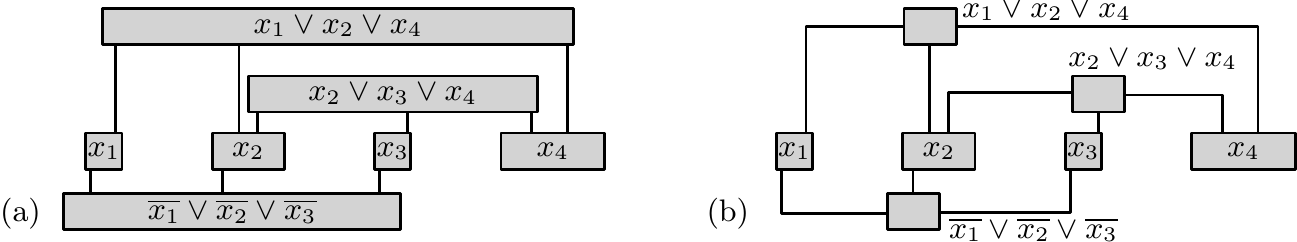}
\caption{(a) Planar monotone 3SAT instance \cite{BergK-2010}. (b) Fixed dimensions for clauses.}
\label{fig:planar3sat}
\end{figure}
\begin{figure}[b]
\centering
\includegraphics{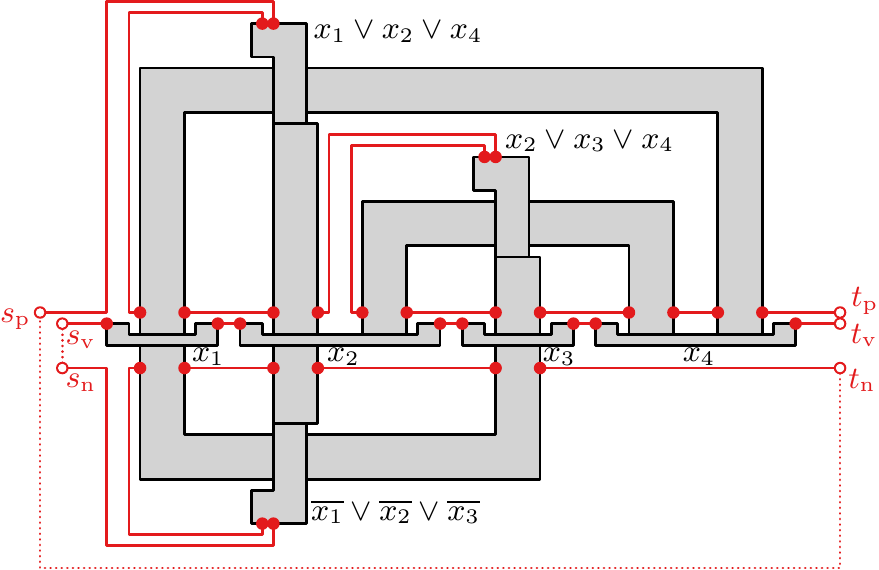}
\caption{Construction sketch for the formula in Fig.~\ref{fig:planar3sat}. Gray polygons represent gadgets, interact ingvia shared boundaries. The red lines connect the various gadgets to obtain a simple polygon.}
\label{fig:constructionoverview}
\end{figure}

For our reduction we construct a simple map matching instance---a partial grid graph $G$ and a polygon $P$---such that $G$ contains a simple cycle $C$ with $\df(C,P) \leq 1$ if and only if formula $F$ is satisfiable.
We use three types of gadgets to represent the variables, clauses and links of $F$: \emph{variable gadgets}, \emph{clause gadgets}, and \emph{propagation gadgets}.
Our clause gadgets have small fixed dimensions and hence cannot be stretched horizontally.
A single bend for the two ``outer'' edges of a clause is sufficient to ensure this (see Fig.~\ref{fig:planar3sat}(b)).
In the upcoming sections we first design the gadgets in isolation before completing the proof using the gadgets.
An overview of the eventual result of this construction is given in Fig.~\ref{fig:constructionoverview}.

\subsection{Gadgets}
\label{ssec:gadgets}

Each gadget specifies a \emph{local graph} (a part of $G$) and a \emph{local curve} (a part of $P$).
The gadgets interact via vertices and edges shared by their local graphs.
There is no interaction based on the local curve: it is used only to force choices in using edges of the local graph.

\mypar{Pressure}
If a cycle exists in the complete graph, a \emph{local path} in the local graph must have a \f distance of at most $1$ to the local curve.
The local path ``claims'' its vertices and edges: these can no longer be used by another gadget.
This results in \emph{pressure} on the other gadget to use a different path.
A gadget has a number of pressure \emph{ports}.
A port corresponds to a sequence of edges in the local graph that may be shared with another gadget.
A port may \emph{receive} pressure, indicating that the shared edges and vertices may not be used in the gadget.
Similarly, it may \emph{give} pressure, indicating that the shared edges and vertices may not be used by any adjacent gadget.
All interaction between gadgets goes via these ports.

\mypar{Gates}
The local curves must be joined carefully to ensure that $P$ is simple.
To this end, each gadget has two curve \emph{gates} that correspond to the endpoints of the local curve.
Later, we show how to connect these gates to create a single simple polygon $P$.

\mypar{Specification}
In the following paragraphs we describe the three gadgets.
In particular, we take note of the \emph{specification} of each gadget: its behavior in terms of its ports; a bounding polygon that contains the local graph and local curve; the placement of its two gates and its ports.
These specifications capture all necessary aspects to complete the reduction in Section~\ref{ssec:proof}.
Here we focus on positive clauses (i.e., above the variables) and their edges.
Gadgets below are defined analogously, by mirroring vertically.

We present specifications and constructions mostly visually, using the following encoding scheme.
The bounding polygon is given with a black outline; the ports are represented with thick green lines; the gates are represented with red dots.
The local graph is given with thick light-blue lines; the local curve is a red line.
We indicate various local paths using dark-blue lines; ports that give pressure for a local path are indicated with an outward arrow.
All elements are visualized on an integer grid (thin gray lines), to show that we indeed construct a partial grid graph.
All coordinates are an integer multiple of a half: all vertices are placed on vertices of the grid, exactly halfway an edge or exactly in the center of a cell.

\mypar{Clause gadget}
A clause gadget is illustrated in Fig.~\ref{fig:clause_gadget}.
It has fixed dimensions; the figure precisely indicates its specification as well as its construction.
The gadget admits a local path only if one of its ports does \emph{not} receive pressure.
Any local path causes pressure on at least one port; for each port there is a path that causes pressure only on that port.
The lack of external pressure on a port indicates that the value of the corresponding variable satisfies the clause.
There is no local path that avoids all three ports: if all ports receive pressure, none of the variables satisfies the clause and the gadget does not admit a local path.

\begin{figure}[b]
\centering
\includegraphics{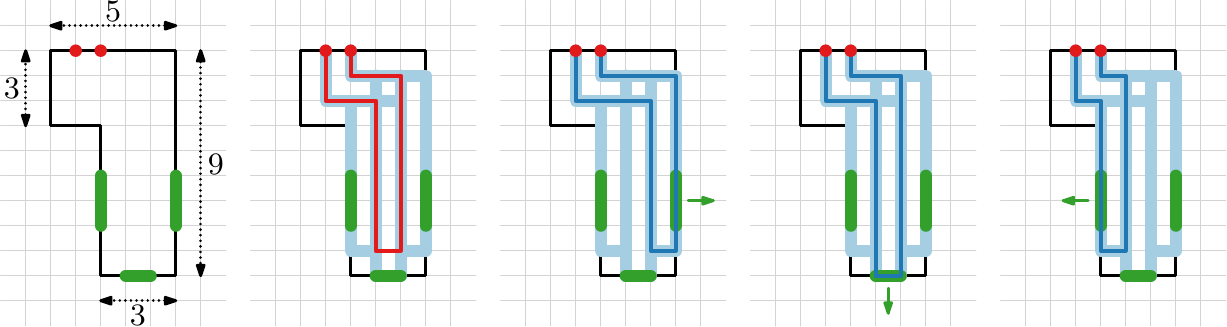}
\caption{The specification (first figure) and construction (second figure) of a clause gadget. Last three figures illustrate local paths; each gives pressure on at least one port.}
\label{fig:clause_gadget}
\end{figure}

\pagebreak
\mypar{Variable gadget}
The specification and construction of a variable gadget depend on the number of literals.
Let $k$ denote the maximum of the number of positive and the number of negative literals of the variable.
We assume that $k > 0$: otherwise, it does not occur in any clause.
Its bounding polygon, gates and ports are illustrated in Fig.~\ref{fig:var_placement} for $k = 2$.
For higher values of $k$, we increase the width by $8$, to ensure a port of width $2$ and a distance of $6$ in between ports.
The gadget admits exactly two local paths: ``\emph{true}'' and ``\emph{false}''.
Ports for positive literals (top side) give pressure only with the \emph{false} path.
Ports for negative literals (bottom side) give pressure only with the \emph{true} path.
In other words, a port gives pressure if the variable does \emph{not} satisfy the corresponding clause.

\begin{figure}[t]
\centering
\includegraphics{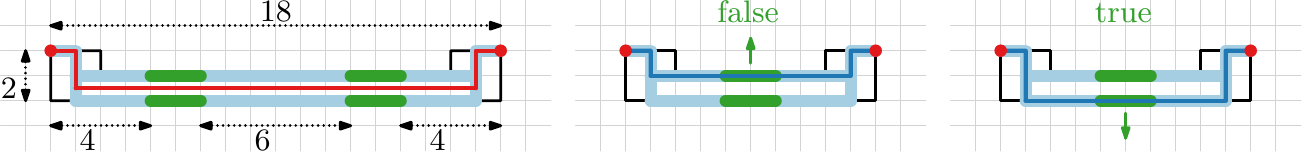}
\caption{The specification and construction of a variable gadget with $k = 2$. The right two figures indicate the two local paths, for $k = 1$.}
\label{fig:var_placement}
\end{figure}

\mypar{Propagation gadget}
A propagation gadget (shown in Fig.~\ref{fig:prop_gadget}) connects a port of a variable gadget to a port of a clause gadget.
The bounding polygon is a corridor of width $4$ with at most one bend: if the link in formula $F$ has a bend, then the gadget also has a bend.
The corridor can have any integer height $h$ greater than $7$.
If it has a bend, the corridor spans any integer width $w$ at least $6$.
The two ports and gates of the gadget are placed as indicated.

\begin{figure}[t]
\centering
\includegraphics{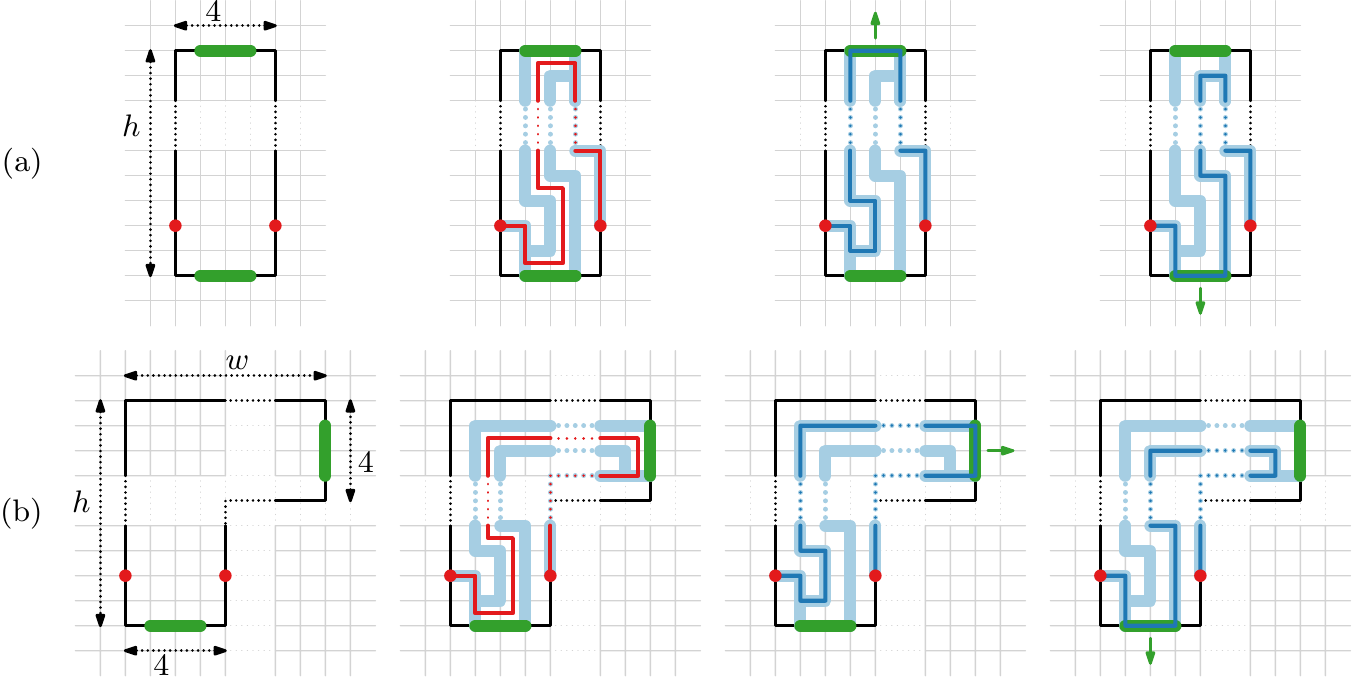}
\caption{The specification (first column) and construction (second column) of a propagation gadget. Last two columns illustrate local paths; each gives pressure on exactly one port. Dotted parts can be stretched to obtain arbitrary dimensions. (a) No bend. (b) Right bend.}
\label{fig:prop_gadget}
\end{figure}

The local graph and curve are constructed such that it admits only two local paths; each puts pressure on exactly one port.
The gadget does not admit a path if both ports receive pressure.
If one port receives pressure, the other must give pressure: it propagates pressure.

\subsection{Construction with gadgets}
\label{ssec:proof}

We are now ready to construct graph $G$ and polygon $P$ based on formula $F$.
Fig.~\ref{fig:constructionoverview} illustrates this construction.
First, we place all variable gadgets next to one another, in the order determined by $F$, with a distance of $2$ in between consecutive variables.

Using the $y$-coordinates in the embedding of $F$, we sort the positive clauses to define a \emph{positive order} $\langle c_1, \ldots c_k \rangle$.
We place the gadget for clause $c_j$ at a distance $7 + 14 j$ above the variables.
Analogously, we use a \emph{negative order} to place the negative clauses below the variables.
Horizontally, the clause gadgets are placed such that the bottom port lines up with the appropriate port on the variable gadget of the middle literal.

Finally, we place a propagation gadget for each link in $F$ to connect the clause and variable gadgets.
By placement of the clauses, any propagation gadget has height at least $7$ and a link without a bend can be represented by a propagation gadget without any bends.
As ports are at least distance $6$ apart, the width of a propagation gadget with a bend exceeds $6$.
A propagation gadget does not overlap other gadgets: the placement of clauses would then imply that the provided embedding for $F$ is not planar.

\mypar{Connecting gadgets}
We have composed the various gadgets in polynomial time.
However, we do not yet have a simple polygon.
We must ``stitch'' the local curves together (in any order) to create polygon $P$.
To this end we first define three subcurves: $P_\text{v}$ for the variable gadgets; $P_\text{p}$ for the positive clause gadgets and their propagation gadgets; and $P_\text{n}$ for the negative clause gadgets and their propagation gadgets. 
Below is a detailed description of how these are constructed. 
Fig.~\ref{fig:constructionoverview} visually illustrates the result.

For $P_\text{v}$ we first define point $s_\text{v}$ and $t_\text{v}$ at distance $4$ outward from the leftmost and rightmost variable gadget respectively, both at the height of the gates.
We connect $s_\text{v}$ to the left gate of the leftmost variable, connect the matching gates of consecutive variables, and then connect the right gate of the rightmost variable to $t_\text{v}$.

Subcurve $P_\text{p}$ is constructed by defining a points $s_\text{p}$ and $t_\text{p}$, similarly as for $P_\text{v}$, but at the height of the ``positive'' propagation gadgets; $s_\text{p}$ is placed at distance $6$ instead.
Analogous to $P_\text{v}$, we first create the straight traversal from $s_\text{p}$ to $t_\text{p}$ through all positive propagation gadgets.
We then include each positive clause in this subcurve, right before it ``enters'' the propagation gadget for its leftmost literal.
This is done by going up starting at distance $3$ before this gadget to two above the top side of the clause gadget, and then connecting to its right gate.
We go back from its left gate, and go down at distance $1$ before the propagation gadget.
Now, $P_\text{p}$ traverses all positive clauses and their propagation gadgets.
Subcurve $P_\text{n}$ is constructed analogous to $P_\text{p}$, though $s_\text{p}$ is placed again at distance $4$.

By placement of the gadgets, these subcurves are simple polygonal curves between their respective endpoints ($s_*$ and $t_*$) and do not intersect each other. 
To obtain a single simple polygon, we must now connect the endpoints of the three subcurves.
We connect $t_\text{v}$ to $t_\text{p}$ and $s_\text{v}$ to $s_\text{n}$ using vertical segments. 
We connect $t_\text{n}$ to $s_\text{p}$, by routing an edge at distance $2$ below the lowest negative clause.

We now have a simple polygon $P$; we define $G$ as the union of all local graphs and the parts of $P$ not contained in a gadget.

\mypar{Proving the theorem}
We now have constructed graph $G$ and polygon $P$.
We must argue that the complexity is polynomial and that $F$ is satisfiable if and only if a simple cycle $C$ exists in $G$ with $\df(C,P) \leq 1$.

Let $n$ denote the number of variables, and $m$ the number of clauses in $F$.
The width of the construction is at most $10 + \sum^n_{i=1} (2 + 8 k_i) + 2 (n - 1)$ where $k_i$ is the number of occurrences of the $i$\textsuperscript{th} variable.
As $\sum^n_{i=1} k_i = 3 m$, the width of the construction is $O(n+m)$.
The height of the construction is at most $2 + (7 + 14 m_\text{p} + 9 + 2) + (7 + 14 m_\text{n} + 9 + 4)$, where $m_\text{p}$ is the number of positive clauses and $m_\text{n}$ the number of negative clauses; since $m_\text{p} + m_\text{n} = m$, the height is $O(m)$.
As all coordinates are required to be an integer multiple of a half, this implies a polynomial bound on the complexity of $G$ and $P$.

Assume that $F$ is satisfiable and consider some satisfying assignment.
We argue the existence of a simple cycle $C$.
For each variable gadget we choose either the ``true'' or ``false'' local path, matching the assigned truth value.
This gives pressure on a number of propagation gadgets: we choose the only remaining local path for these, causing pressure on the corresponding clauses.
For the other propagation gadgets, we choose the path such that it gives pressure at the variable and may receive pressure at the clause.
Since the truth values of the variables originate from a satisfying assignment, at most two ports of any clause receive pressure.
Hence, the clause admits a local path as well.
We concatenate the local paths with the paths that are used to stitch together the local curves to obtain a simple cycle $C$.
By construction $\df(C,P)$ is at most $1$.

Now, assume that $G$ contains a simple cycle $C$ with $\df(C,P) \leq 1$.
By construction, cycle $C$ traverses all gadgets and contains exactly one local path for each gadget.
This local path ends at the gates of the gadget and the \f distance between this local path and the local curve is at most $1$.
For a variable, this local path corresponds to either the true or false state.
This directly yields the truth values of the variables.
Each clause gadget also has a local path and hence one or more of its ports give pressure.
Since the propagation gadgets have a local path, the pressure from the clauses results in pressure on a variable gadget.
This pressure ensures that a variable that receives pressure from a clause is in a state satisfying the clause.
Hence, the truth values found from the variables yield a satisfying assignment for formula $F$.

This concludes the proof for Theorem~\ref{thm:main}, showing that simple map matching is indeed NP-complete.
This result and its construction have a number of implications, as discussed below.

\subsection{Implications}
\label{sec:corollaries}

\myparNS{Approximation}
With the NP-hardness result above, we may want to turn to approximation algorithms.
However, a simple argument shows that approximation is also NP-hard.

\begin{corollary}\label{cor:noapprox}
Let $G$ be a partial grid graph of complexity $m$ and let $P$ be a simple polygon of complexity $n$.
It is NP-hard to approximate the minimal $\df(C,P)$ of any simple cycle $C$ in $G$ within any factor of $2^{\textrm{\emph{poly}}(n,m)}$ where $\textrm{\emph{poly}}(n,m)$ is a polynomial in $n$ and $m$.
\end{corollary}
\begin{proof}
For an unsatisfiable formula the minimal \f distance of a simple cycle in the constructed graph is significantly larger than $1$.
Suppose that this minimal \f distance is strictly greater than $c$.
Any $c$-approximation algorithm for simple map matching is able to decide satisfiability of planar monotone 3SAT formulas.
Thus, unless P$=$NP, no $c$-approximation algorithm can have polynomial execution time.

To determine the exact value of $c$, we wish to determine at which value a local path becomes admissible that does not correspond to the desired behavior of the construction.
This occurs at $c = 6$: the clause gadget breaks, admitting a local path that does not pass through any of the ports.
However, it is straightforward to lengthen the gadgets to increase the value of $c$.
This does not increase the combinatorial complexity of $G$ and $P$, thus maintaining a polynomial-size instance.
The construction now spans an area $O(c(n_F+m_F))$ by $O(cm_F)$, where $n_F$ and $m_F$ are the number of variables and clauses in formula $F$.
Thus, we require $c \leq 2^{\textrm{poly}(m_F,n_F)}$ to encode the coordinates in polynomial space.
\end{proof}

\begin{wrapfigure}[12]{r}{0.3\linewidth}
  \centering
  \includegraphics{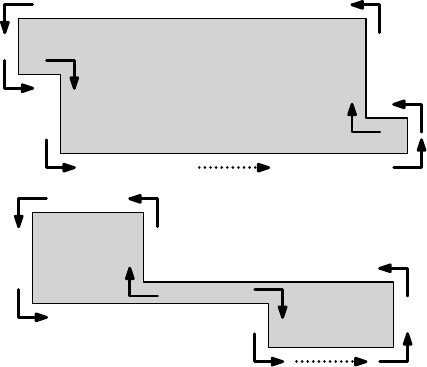}
  \caption{Two polygons with the same bend profile.}
  \label{fig:turnprofile}
\end{wrapfigure}
\myparNS{Counting bends}
For schematization we do not wish to just find some cycle in $G$, but also optimize or bound its complexity, measured in the number of bends.
If $G$ is a partial grid graph (corresponding to rectilinear schematization), any bend is either a left or right turn.
Every rectilinear polygon has a certain \emph{bend profile}: the sequence of left and right bends in counterclockwise order along its boundary.
The bend profile gives no information about edge lengths: seemingly different polygons have the same bend profile (see Fig.~\ref{fig:turnprofile}).
Unfortunately, even with a given bend profile, the problem remains NP-complete.
However, in this reduction the bend profile has length proportional to the complexity of the formula.
It does not prove that no fixed-parameter-tractable (FPT) algorithm exists.

\begin{corollary}\label{cor:nobends}
Let $G$ be a partial grid graph, let $P$ be a simple polygon, let $\eps > 0$ and let $T$ be a bend profile.
It is NP-complete to decide whether $G$ contains a simple cycle $C$ that adheres to $T$ with $\df(C,P) \leq \eps$.
\end{corollary}
\begin{proof}
In all constructions, the bends made by the various local paths are identical.
We can easily derive a bend profile that must lead to a simple cycle, if the formula is satisfiable.
\end{proof}

\myparNS{Area preservation}
Suppose we want an area-preserving solution, i.e., the area of cycle $C$ must be equal to that of $P$. 
A simple argument proves that our reduction can be extended to prove that this more constrained problem is also NP-complete. 

\begin{corollary}\label{cor:smm-area}
Let $G$ be a partial grid graph, let $P$ be a simple polygon and let $\eps > 0$.
It is NP-complete to decide whether $G$ contains a simple cycle $C$ with $\df(C,P) \leq \eps$ and $|C| = |P|$.
\end{corollary}
\begin{proof}
Fig.~\ref{fig:constructionoverview_complement} at the end of the paper shows an overview of the modified construction.
The key observation is that the original construction is duplicated, but inside and outside the polygon have been inverted.
This is achieved by connecting the endpoints of the subcurves $P_v$, $P_p$ and $P_n$, slightly differently, i.e., with those of the duplicate.
Any solution $C$ coincides with $P$ outside the gadgets: only the local paths change the area of $C$. 
Hence, any change in area resulting from a local path in one of the gadgets in the one copy can be counteracted by choosing the exact same local path in the other copy.
\end{proof}

\myparNS{Variants}
Finally, there are a number of variants of the problem that can be proven to be NP-complete via the same construction.
As strict monotonicity in the homeomorphism is not crucial for the reduction, the problem under the \emph{weak} \f distance is also NP-complete.
The clause gadget admits an extra local path, one that still exhibits the desired behavior.
The problem is also NP-complete under the \emph{discrete} (weak) \f distance, as we may sample the graph and polygon appropriately.
Not every grid location needs a vertex, preserving the inapproximability result of Corollary~\ref{cor:noapprox}. 
As all interaction between gadgets is based on edges, it is also NP-complete to determine the existence of an ``edge-simple'' cycle that uses each edge at most once (but vertices may be used more than once).
It is not essential in any of the above (except for area preservation) for $P$ to be a closed curve: the same construction works for open curves, and thus these variants are NP-hard as well.

\section{Connected face selection is NP-complete}
\label{sec:cfs}

With the negative results using the \f distance, we now consider the symmetric difference via Problem~\ref{prob:CFS}, connected face selection. We prove that this problem is also NP-complete, even on a full grid graph, as captured in the following theorem.
\begin{theorem}\label{thm:cfs}
Let $G$ be a full grid graph, let $P$ be a simple polygon and let $D > 0$.
It is NP-complete to decide whether $G$ contains a connected face set $S$ with $\ds(S,P) \leq D$.
\end{theorem}
The problem is obviously in NP, since the symmetric difference and connectedness can be straightforwardly verified in polynomial time.
Here we show that it is also NP-hard.
We conclude this section with some implications of this result.

\subsection{Reduction}

\myparNS{Rectilinear Steiner tree}
The rectilinear Steiner tree problem is formulated as follows:
given a set $X$ of $n$ points in $\R^2$, is there a tree $T$ of total edge length at most $L$ that connects all points in $X$, using only horizontal and vertical line segments?
Vertices of $T$ are not restricted to $X$.
This problem was proven NP-complete by Garey and Johnson~\cite{GareyJ-1977}.
Hanan~\cite{Hanan-1966} showed that an optimal result must be contained in graph $H(X)$ corresponding to the arrangement of horizontal and vertical lines through each point in $X$ (see Fig.~\ref{fig:cfs_sketch}(a)). Subsequently, this was called the \emph{Hanan grid} (e.g. \cite{Zachariasen-2001}). 
We call a vertex of $H(X)$ a \emph{node} if it corresponds to a point in $X$ and a \emph{junction} otherwise.
As the problem is scale invariant, we assume $L = 1$. 
All edges in $H(X)$ must be shorter than $1$: otherwise, the answer is trivial---no such tree exists.

\begin{figure}[t]
\centering
\includegraphics{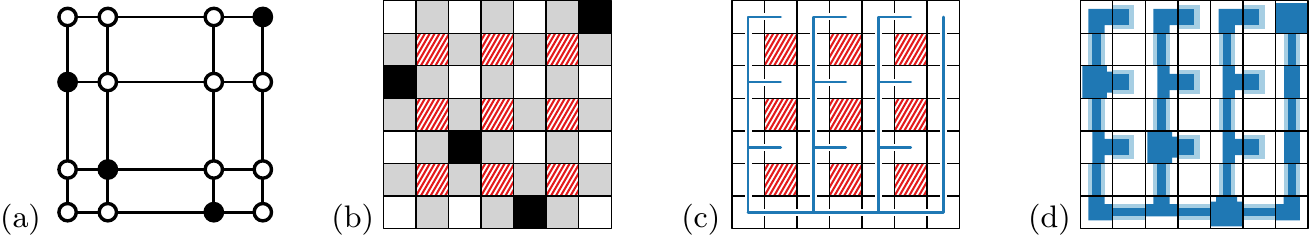}
\caption{Sketch of the reduction. (a) The Hanan grid of $X$ (dark points). (b) $G$ consists of unit square cells: node-cells are colored black, junction-cells white, edge-cells gray; face-cells are hatched. (c) ``Skeleton'' $\varsigma$ for polygon $P$. (d) Sketch of polygon $P$ constructed on top $G$.; light blue areas indicate flexible parts in the construction that represent edge lengths.}
\label{fig:cfs_sketch}
\end{figure}

\mypar{Reduction overview}
We must transform point set $X$ into a full grid graph $G$, a polygon $P$ and a value $D > 0$.
We construct $G$ such that each \emph{cell} (face of $G$) corresponds to a vertex (\emph{node-cell} or \emph{junction-cell}), an edge (\emph{edge-cell}), or a bounded face (\emph{face-cell}) of $H(X)$; see Fig.~\ref{fig:cfs_sketch}(b). 
We then construct polygon $P$ by defining a part of the polygon inside all cells, except for the face-cells: $P$ does not overlap these.
To structure $P$ we use a \emph{skeleton} $\varsigma$, a tree spanning the non-face-cells in the dual of $G$ (Fig.~\ref{fig:cfs_sketch}(c)).

\mypar{Weights}
Recall that the symmetric difference between $P$ and a face set $S = \{ c_1, c_2, \ldots \}$ may be computed as $|P| + \sum_{c \in S} (|c| - 2 \cdot |P \cap c|)$.
As $|c| = 1$, we define the \emph{weight} of a cell $c$ in $G$ as $w(c) = 1 - 2 \cdot |P \cap c|$.
Hence, the symmetric difference is $|P| + \sum_{c \in S} w(c)$.
We set the desired weight $w(c)$ for cell $c$ to: 
%
\begin{itemize}
\item \makebox[2.5em][l]{$-\frac{3}{4}$} if $c$ is a node-cell; 
\item \makebox[2.5em][l]{$0$} if $c$ is a junction-cell; 
\item \makebox[2.5em][l]{$\|e\|/2$} if $c$ is an edge-cell, where $\|e\|$ is the length of the corresponding edge $e$ in $H(X)$;
\item \makebox[2.5em][l]{$1$} if $c$ is a face-cell.
\end{itemize}
Given a desired weight $w(c)$ for cell $c$, the area of overlap $A(c)$ is $|P \cap c| = \frac{1 - w(c)}{2}$.
Every cell has $A(c)$ area of $P$ inside; $|P|$ equals $\sum_{c \in G} A(c)$.
We call $P \cap c$ the \emph{local polygon} of $c$.
We set $D = |P| - \frac{3}{4} n + \frac{1}{2} = (\sum_{c \in G} A(c)) - \frac{3}{4} n + \frac{1}{2}$.
That is, the sum of weights is at most $-\frac{3}{4} n + \frac{1}{2}$.
This can be achieved only if the face set contains all node-cells and no face-cell.

\mypar{Designing cells}
We design every cell such that the desired weight is achieved. 
For a face-cell, this is trivial: $w(c) = 1$, hence $A(c) = 0$ and we keep $P$ disjoint from this cell. 
For all other cells $P$ to cover some fraction of its interior, as dictated by $A(c)$. 
Skeleton $\varsigma$ dictates how to connect the local polygons; we ensure that at least the middle $25\%$ of the shared edge (the \emph{connector}) is covered. 
A local polygon should never touch the corners of its cell.

Node- and junction-cells may have up to four neighbors in $\varsigma$. 
Covering the four connectors is done with a cross shape, covering $\frac{7}{16}$th of the cell's area. 
By thickening this cross, we can straightforwardly support $A(c)$ with $\frac{7}{16} \leq A(c) < 1$, ensuring not to touch the corners or edges shared with cells that are not adjacent in $\varsigma$.
A node-cell has weight $-\frac{3}{4}$; $A(c) = \frac{7}{8}$. 
A junction-cell has weight $0$; $A(c) = \frac{1}{2}$. 
Hence, both can be represented (see Fig.~\ref{fig:cfs_design}(a--b)).

An edge-cell has weight $\|e\|/2$ and thus should cope with weights between zero and a half; $A(c)$ lies between $\frac{1}{4}$ and $\frac{1}{2}$. 
Any edge-cell has degree $1$ or $2$ in $\varsigma$; if it has degree $2$, the neighboring cells are on opposite sides.
Hence, $A(c) = \frac{1}{4}$ can be trivially handled by creating a rectangular shape that touches exactly the necessary connectors.
$A(c) = \frac{1}{2}$ is dealt with by widening this rectangle within the cell; any intermediate weight is handled by interpolating between these two.
This is illustrated in Fig.~\ref{fig:cfs_design}(c).

\begin{figure}[t]
\centering
\includegraphics{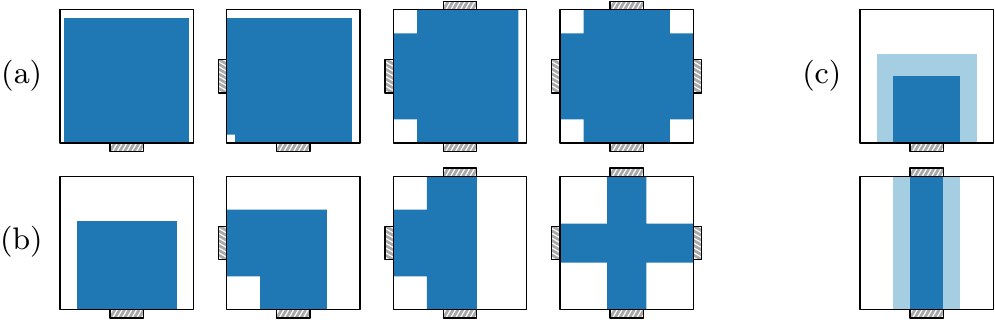}
\caption{Examples of cell design, with connectors as hatched rectangles. (a) Node-cells are covered for $87.5\%$ by $P$. (b) Junction-cells are covered for $50\%$ by $P$. (c) Between $25\%$ (dark blue) and $50\%$ (dark and light blue) of edge-cells are covered by $P$.}
\label{fig:cfs_design}
\end{figure}

\mypar{Proving the theorem}
We now have graph $G$, polygon $P$ and a value for $D$.
We must prove that the reduction is polynomial and that $X$ has a rectilinear Steiner tree of length at most $1$, if and only if there is a connected face set $S$ in $G$ such that $\ds(S,P) \leq D$.
The former is trivial by observing that: $P$ has $O(1)$ complexity in each of the $O(n^2)$ cells; computing tree $T$ for guiding the construction of $P$ can be done with for example a simple breadth-first search.

Suppose we have a rectilinear Steiner tree $T$ of length at most $1$ in the Hanan grid.
We construct a face set $S$ as the union of all cells corresponding to vertices and edges in $T$.
By definition of $T$, this must contain all node-cells and cannot contain face-cells.
As junction-cells have no weight, the total weight of $S$ is $-\frac{3}{4} n + \sum_{e \in T} w(c_e) = -\frac{3}{4} n + \frac{1}{2} \sum_{e \in T} \| e \|$ where $c_e$ is the cell of $G$ corresponding to edge $e$.
By assumption $\sum_{e \in T} \| e \|\leq 1$: the total weight is at most $-\frac{3}{4} n + \frac{1}{2}$.
Thus, the symmetric difference for $S$ is at most $|P| -\frac{3}{4} n + \frac{1}{2} = D$.

Suppose we have a connected face set $S$ in $G$ such that $\ds(S,P) \leq D$.
The total weight is thus $D - |P| = -\frac{3}{4} n + \frac{1}{2}$. 
Since face-cells have weight $1$ and only node-cells have negative weight, being $-\frac{3}{4}$, this can be achieved only if $S$ contains all node-cells and no face-cells.
In particular, the sum of the weights over all edge-cells is at most $\frac{1}{2}$.
Thus, the subgraph of $H(X)$ described by the selected cells must be connected, contain all nodes of $X$, and have total length at most $1$.
If this subgraph is not a tree, we can make it a tree, by leaving out edges (further reducing the total length), until the subgraph is a tree.

\subsection{Implications} 

\myparNS{Simply connected}
The same reduction works for a \emph{simply} connected face set, as Steiner tree $T$ cannot contain cycles and a simply connected face set in $G$ readily implies a tree.

\begin{corollary}
Let $G$ be a full grid graph, let $P$ be a simple polygon and let $D > 0$.
It is NP-complete to decide whether $G$ contains a simply connected face set $S$ with $\ds(S,P) \leq D$.
\end{corollary}
\begin{proof}
The face set obtained from a Steiner tree must be simply connected, since $T$ cannot contain cycles. 
In addition, a simply connected face set in $G$ (that does not contain a face-cell) directly describes a tree.
\end{proof}

\myparNS{Area preservation}
With an area-preservation constraint the problem remains NP-complete. 
We only sketch an argument; full proof can be found in Appendix~\ref{app:cfs-impl}.
This readily implies that variants prescribing the number of faces or total area via a parameter are also NP-hard.

\begin{corollary}\label{cor:cfs-area}
Let $G$ be a full grid graph, let $P$ be a simple polygon and let $D > 0$.
It is NP-complete to decide whether $G$ contains a (simply) connected face set $S$ with $\ds(S,P) \leq D$ and $|S| = |P|$.
\end{corollary}
\begin{proof}
We assume for simplicity that no points in $X$ share an $x$- or $y$-coordinate: $H(X)$ is then exactly an $n \times n$ grid.
The number of cells we need to represent a Steiner tree spanning $X$ is at least $2 n - 1$ and at most $2 n^2 - 1$.
Thus, we need $|P| = \sum A(c) \geq 2 n^2 - 1$, as to allow sufficiently many cells to be selected in the original construction. 

To achieve the necessary area for $P$, we are going to add more cells to the construction, each with $A(c) = \frac{1}{2}$ and thus $w(c) = 0$.
We add $(2 n - 1)^2 + 1 = 4n^2 - 4n + 2$ such cells.
Due to their weight, their presence in or absence from face set $S$ has no effect on the symmetric difference.
We need the new cells to not interfere with the original construction.
Therefore, they should be adjacent only with one of the node-cells on the boundary of the construction.
Thus we also add $2 n - 2$ cells weight $w(c) = 1$ to separate the new cells from the original cells.
Fig.~\ref{fig:CFS_sketcharea} illustrates an overview of this extended construction and its new skeleton.
Note that the newly added cells have the same weights as the original junction-cells and face-cells; the construction of $P$ thus extends straightforwardly.

\begin{figure}[h]
\centering
\includegraphics{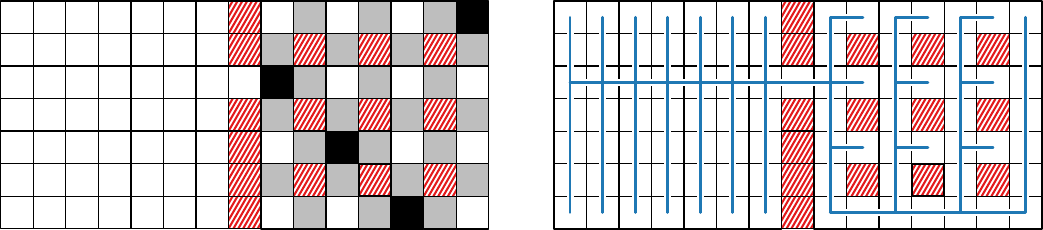}
\caption{Adding overflow cells to the construction, cell types (left) and skeleton (right).}
\label{fig:CFS_sketcharea}
\end{figure}

If we are given some connected face set (that now uses a fixed number of cells) that has sufficiently low symmetric difference, we still know that it must span all the node-cells and thus encode a rectilinear Steiner tree of sufficiently low length.

To transform a rectilinear Steiner tree $T$ to a connected face set, we must argue that we can always select the correct number of cells.
First, let us bound the size of $|P|$. 
Since edges have $A(c)$ in between $\frac{1}{4}$ and $\frac{1}{2}$, it is easy to derive that the total area $A_o$ of the original construction satisfies $\frac{7}{8} n + \frac{1}{2} (n^2 - n) + \frac{1}{4} \cdot 2 n (n - 1) \leq A_o \leq \frac{7}{8} n + \frac{1}{2} (n^2 - n) + \frac{1}{2} \cdot 2 n (n - 1) $; this simplifies to $n^2 - \frac{1}{8} n \leq A_o \leq \frac{3}{2} n^2 - \frac{5}{8} n$.
The new cells add $A_n = \frac{1}{2} (4n^2 - 4n + 2) = 2n^2 - 2n + 1$ area.
Thus, $|P| = A_n + A_o$ is bounded by the interval 
$[3n^2 - \frac{17}{8} n + 1, \frac{7}{2} n^2 - \frac{21}{8} n + 1]$.
Hence, the number of cells we need to select is more than $2n^2 -1$ and strictly less than the number of newly added cells with $w(c) = 0$.
Therefore, we simply apply the original transformation, but use the new cells as ``overflow'' for any cells in excess of those needed to represent $T$.

The above assumes that the eventual area of $P$ is an integer; since it depends on the edge lengths in $H(X)$, it likely is not.
This can be remedied by either assuming we are allowed to round the area of $P$, or by letting $P$ extend slightly into the outer face to make it integer. 
Note that the weight of the outer face is always infinite for a bounded $P$.
Alternatively, we can use any weight for node-cells, strictly in between $\frac{1}{2}$ and $1$.
\end{proof}

\mypar{Variants}
More general graph classes (e.g. plane graphs) are also NP-hard; this is readily implied by Theorem~\ref{thm:cfs}. 
Finally, the problem remains NP-complete for graphs representing hexagonal and triangular tilings (by combining two triangles into one slanted square).

\newpage
\section{Conclusions}
\label{sec:conclusion}

Schematic maps are an important tool to visualize data with a geographic component.
We studied discretized approaches to their construction, by restricting solutions to a plane graph $G$.
This promotes alignment and uniformity of edge lengths and avoids the risk of visual collapse.
We considered two variants: \emph{simple map matching} using the \f distance and \emph{connected face selection} using the symmetric difference.
Unfortunately, both turn out to be NP-complete; the former is even NP-hard to approximate.
The proofs readily imply that a number of variants are NP-hard as well, even with an area-preservation constraint.


\mypar{Open problems}
It remains open whether (general) simple map matching is fixed-parameter tractable. 
Moreover, there is quite a gap between the graph necessary for the reduction and the graphs that would be useful in the context of schematization (e.g. full grid graphs, triangular tilings).
Are such instances solvable in polynomial time?
For connected face selection, we know hardness on such constructed graphs, but it remains open whether approximation or FPT algorithms are possible.
In both methods the reduction needs rather convoluted polygons, very unlike the geographic regions that we want to schematize: do realistic input assumptions help to obtain efficient algorithms?
Recently, Bouts~\etal~\cite{BoutsKKMSV-2016a,BoutsKKMSV-2016b} have show that any full grid graph admits a simple cycle with \f distance bounded using a realistic input model, called narrowness. 
However, this does not readily preclude the decision problem from being NP-hard, even for polygons that have bounded narrowness.

The \f distance is a bottleneck measure and thus results obtained via simple map matching may locally deviate more than necessary, even when minimizing the number of bends.
Buchin~\etal~\cite{BuchinBMS-2012} introduced ``locally correct \f matchings'' to counteract this flaw with the \f distance.
Can we extend this concept to simple map matching?

\mypar{Acknowledgments}
{
The author would like to thank: Bettina Speckmann and her Applied Geometric Algorithms group, Kevin Buchin, Bart Jansen, Arthur van Goethem, Marc van Kreveld and Aidan Slingsby for inspiring discussions on the topic of this paper; Gerhard Woeginger for pointing out the exponential bound in Corollary~\ref{cor:noapprox}. 
The author was partially supported by the Netherlands Organisation for Scientific Research project 639.023.208 and Marie Sk\l{}odowska-Curie Action MSCA-H2020-IF-2014 656741.
}

\begin{figure}
\centering
\includegraphics{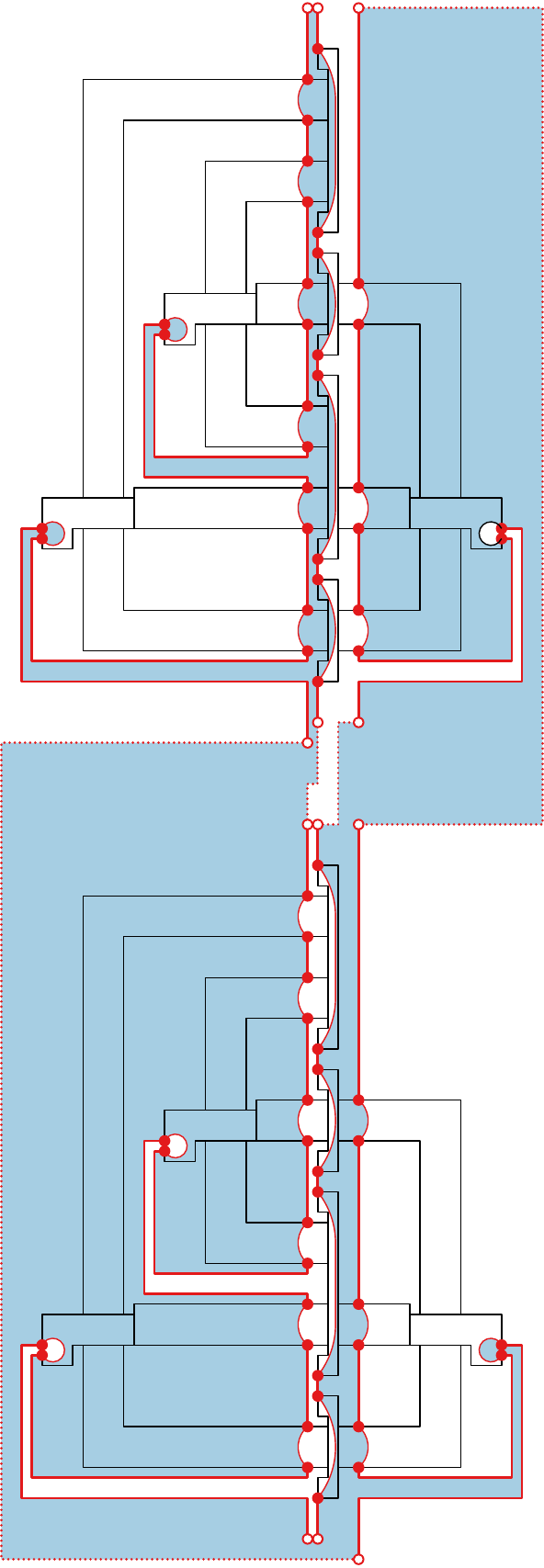}
\caption{Duplication strategy allows us to ensure that an area-preserving solution must exist for a satisfying formula. Local curves have been schematically represented by a circular arc. The interior of $P$ is shaded. Construction has been rotated by 90 degrees.}
\label{fig:constructionoverview_complement}
\end{figure}

\end{document}